\documentclass[11pt]{article}
    
    \usepackage[margin=1in]{geometry}
    \usepackage{amsmath,amssymb,amsthm,mathtools}
    \usepackage{microtype}
    \usepackage{enumitem}
    \usepackage{xspace}
    \usepackage[colorlinks=true,linkcolor=blue,citecolor=blue,urlcolor=blue]{hyperref}
    \usepackage{comment}
    
    \newtheorem{theorem}{Theorem}[section]
    \newtheorem{lemma}[theorem]{Lemma}
    \newtheorem{proposition}[theorem]{Proposition}
    \newtheorem{corollary}[theorem]{Corollary}
    \newtheorem{definition}[theorem]{Definition}

    \newcommand{\F}{\mathbb{F}}
    \newcommand{\E}{\mathbb{E}}
    \newcommand{\Tr}{\operatorname{Tr}}
    \newcommand{\wt}{\operatorname{wt}}
    \newcommand{\poly}{\operatorname{poly}}
    
    \newcommand{\cP}{\mathcal{P}}
    \newcommand{\LPC}{\textnormal{\textsc{LPC}}\xspace}
    \newcommand{\GapMWC}{\textnormal{\textsc{GapMWC}}\xspace}
    
    \newcommand{\Yes}{\textnormal{\textsc{Yes}}\xspace}
    \newcommand{\No}{\textnormal{\textsc{No}}\xspace}
    \newcommand{\ket}[1]{|#1\rangle}
    \newcommand{\bra}[1]{\langle #1|}
    \newcommand{\braket}[2]{\langle #1|#2\rangle}
    
    \newcommand{\BQP}{\textsc{BQP}}
    \newcommand{\Pclass}{\textsc{P}}
    \newcommand{\NP}{\textsc{NP}}
    \newcommand{\QCMA}{\textsc{QCMA}}
    
    \title{Complexity of detecting large coefficients in the Pauli basis}
    \author{Santiago Cifuentes\\
    ICC Conicet, Universidad de Buenos Aires}
    \date{June 2026}
    
\begin{document}
    \maketitle
    
    \begin{abstract}
    We study the problem of deciding, given a mechanism to prepare a quantum state $\rho$ and a value $\varepsilon > 0$, whether there is some non-identity Pauli matrix $P$ such that $|\Tr(P \rho)| \geq \varepsilon$. We consider that the state $\rho$ is described as the result of tracing out some of the qubits of a pure state prepared by a circuit $C$, and we assume the promise that either there is a Pauli matrix satisfying the stated condition or, instead, that for all non-identity Pauli matrices $P$ it is the case that $|\Tr(P\rho)|\leq \varepsilon/2$. The problem is in \QCMA{}, and we prove that if it belongs to \BQP{} then $\NP{} \subseteq \BQP{}$. The result is obtained through a reduction from the minimum-weight code problem, and it holds even when $\rho$ is assumed to be a pure state (i.e. when no qubits are discarded) and $\varepsilon$ is constant. This resolves an open question regarding the existence of efficient tomographic procedures to find the largest coefficients of a quantum state in the Pauli basis: namely, they do not exist under the standard hypothesis $\NP{} \nsubseteq \BQP{}$.
    \end{abstract}
    
    \section{Introduction}
    \label{sec:introduction}
    
    In this paper we address the following problem: given a succinct preparation of a state $\rho$,
    decide whether there exists a non-identity Pauli with expectation value on $\rho$ above some threshold. The
    problem is immediately in QCMA when the threshold has a reasonable size: the classical witness is
    a description of the Pauli matrix, and the verifier estimates the corresponding expectation value by repeated
    Pauli measurements. The main result of this work is that this problem is also, in some sense, $\NP{}$-hard, even when $\rho$ is assumed to be pure and the threshold constant (i.e. independent of the input). More precisely, we show that the existence of a polynomial-time quantum
    algorithm for this problem would imply that $\mathrm{NP}\subseteq\mathrm{BQP}$. Thus, under the common theoretical assumption that $\NP{}\nsubseteq\BQP{}$, we expect that this problem is not in $\BQP{}$.
    
    This problem is closely related to selective and partial forms of quantum tomography, and in previous literature algorithms were developed to estimate specific Pauli coefficients efficiently. In particular, \cite{bendersky2008selective} and~\cite{bendersky2013selective} introduced procedures for estimating parameters of a quantum process in the Pauli basis, and~\cite{bendersky2009general} studied two-copy measurement schemes for finding some of these parameters. More recently, \cite{ciancaglini2026measuring} introduced a hierarchical algorithm for identifying the largest Pauli coefficients of an unknown state, which runs in exponential time in the worst case. This problem was also mentioned recently in the context of Pauli-sampling~\cite{hinsche2025efficient, hangleiter2024bell}, where it was highlighted as an open and relevant problem \cite{hinsche2025efficient}. Other results related to computing expectation values of Pauli matrices show that in restricted scenarios few measurements are enough to recover relevant observables \cite{huang2021efficient}, and more generally, shadow-tomography methods can estimate many prescribed observables using polynomially many copies, although the classical postprocessing may be exponential \cite{aaronson2018shadow, huang2020predicting}. We note that, in general, it is not known whether it is possible to detect the presence of large coefficients in the Pauli basis using a polynomial amount of both quantum and classical resources.
    
    Our proof is obtained through a reduction from a gap version of the minimum-weight nonzero codeword problem~\cite{vardy1997intractability}. Essentially, this problem consists in deciding, given a generator matrix
    $G\in\F_2^{r\times m}$ and a threshold $t$, whether there is some nonzero $x \in \F_2^r$ such that $\wt(xG) \leq t$, where $\wt{}$ denotes the Hamming weight. The core observation for the reduction is that, if $g_1,\ldots,g_m \in \F_2^r$ are the columns of $G$, it holds that
    \[
        \sum_{j=1}^m(-1)^{x\cdot g_j}=m-2\wt(xG).
    \]
    The reduction constructs a density matrix whose relevant Pauli coefficients are monotone functions of $\wt(xG)$ such that a codeword with small Hamming weight corresponds to a large Pauli coefficient.
    
    The rest of the paper is organized as follows. In Section~\ref{sec:preliminaries} we fix notation, define the circuit model for mixed states, and state precisely the two promise problems involved in the reduction. In Section~\ref{sec:basic-lemmas} we prove the elementary identities connecting diagonal quantum states, Pauli coefficients, and Hamming weights of codewords, as well as a small lemma required to extend the proof to the pure case. In Section~\ref{sec:main-results} we prove that \LPC (\textit{Large Pauli Coefficient}) belongs to \QCMA{}, give the reduction from the gap minimum-weight codeword problem to \LPC, and finally show that \LPC{} is also \BQP{}-hard. In Section~\ref{sec:pure state-case} we improve the reduction to prove hardness for pure states. Finally, in Section~\ref{sec:discussion} we provide a conclusion summarizing the main takeaways.
    
    \section{Preliminaries}
    \label{sec:preliminaries}
    
    Let
    \[
        \cP_n=\{I,X,Y,Z\}^{\otimes n}
    \]
    denote the phase-free $n$-qubit Pauli strings. For $x=(x_1,\ldots,x_n)\in\F_2^n$ we let
    \[
        Z^x:=Z^{x_1}\otimes\cdots\otimes Z^{x_n},
    \]
    where $Z^0=I$ and $Z^1=Z$.
    
    We will assume quantum states are described as the result of tracing out part of a pure state obtained through a quantum circuit. We formalize this input model as follows.
    
    \begin{definition}
    A state preparation circuit $C$ is a quantum circuit with all input
    qubits initialized to $\ket{0}$. The circuit has a designated output register $A$ of $n$
    qubits and an auxiliary register $B$. If the final pure state on $AB$ is $\rho_{AB}$, the
    state prepared by $C$ is
    \[
        \rho_C:=\Tr_B(\rho_{AB}),
    \]
    where $\Tr_{B}$ denotes the partial trace over the register $B$.
    \end{definition}
    
    We now define formally the problem we address.
    
    \begin{definition}
    An instance of \LPC (i.e. \textit{Large Pauli Coefficient}) is a pair $(C,\varepsilon)$, where $C$ is a state preparation circuit
    for an $n$-qubit state $\rho_C$, and $\varepsilon\in(0,1]$ is a threshold. The task consists in deciding whether there is some non-identity Pauli matrix whose expectation value over $\rho_C$ has absolute value at least $\varepsilon$, under the promise that if there is no such Pauli matrix then all coefficients have absolute value at most $\varepsilon/2$. More formally, an algorithm that solves the problem must return, given $(C, \varepsilon)$,
    \[
    \begin{array}{lll}
    \Yes & \text{if} & \exists P\in\cP_n\setminus\{I^{\otimes n}\}\text{ such that }
            |\Tr(P\rho_C)|\geq\varepsilon,\\[1mm]
    \No & \text{if} & \forall P\in\cP_n\setminus\{I^{\otimes n}\},\quad
            |\Tr(P\rho_C)|\leq\varepsilon/2.
    \end{array}
    \]
    As usual, we assume that one of the cases holds.
    \end{definition}
    
    We will show that if $\LPC \in \BQP{}$ then $\NP{} \subseteq \BQP{}$. This result is obtained through a reduction from the minimum-weight codeword problem, which we now introduce.
    
    For a binary vector
    $u=(u_1,\ldots,u_m)\in\F_2^m$ we define its Hamming weight as
    \[
        \wt(u):=\bigl|\{j\in\{1,\ldots,m\}:u_j=1\}\bigr|
    \]
    and we denote the inner product between $x,y\in\F_2^r$ as
    \[
        x\cdot y:=\sum_{i=1}^r x_i y_i \pmod 2.
    \]
    We will often omit the multiplication symbol and write $xy$ to denote $x \cdot y$.
    
    Let $G\in\F_2^{r\times m}$ be a binary matrix. We view its row span as a binary linear code
    \[
        \mathcal{C}(G):=\{xG:x\in\F_2^r\}\subseteq\F_2^m.
    \]
    The minimum nonzero codeword weight is
    \[
        \Delta(G):=\min\{\wt(c):c\in \mathcal{C}(G),\ c\neq 0\}.
    \]
    Here we naturally assume that $\mathcal{C}(G) \neq \{0^m\}$. We will consider a promise version of the problem of computing $\Delta(G)$ introduced in~\cite{dumer2003hardness}[Definition 1].
    
    \begin{definition}
    Fix a constant $\lambda>1$. An instance of $\lambda$-\GapMWC is a pair $(G,t)$, where
    $G\in\F_2^{r\times m}$ is a binary matrix and $t$ is a positive integer. An algorithm that solves this problem must return
    \[
    \begin{array}{lll}
    \Yes & \text{if} & \Delta(G)\leq t,\\[1mm]
    \No & \text{if} & \Delta(G)>\lambda t,
    \end{array}
    \]
    assuming that one of the cases holds.
    \end{definition}
    
    If $G$ has full row rank then $xG = 0$ if and only if $x = 0$. Thus, it holds that
    \[
        \Delta(G)=\min_{0\neq x\in\F_2^r}\wt(xG).
    \]
    If $G$ is not full row rank, we can perform Gaussian elimination in polynomial time and replace it with a full row rank matrix
    with the same span. Thus, from now on we assume that the input of \GapMWC is a full row rank, non-trivial matrix. Without loss of generality we will also assume that $1 \leq t <m$.
    
    The exact minimum distance problem for binary linear codes was shown to be NP-hard by
    Vardy~\cite{vardy1997intractability}. Constant-factor gap hardness was proved in~\cite{dumer2003hardness} under randomized reductions\footnote{Strictly speaking, what is proven is that an algorithm solving the gapped version would imply that $\textsc{RP} = \NP$: the promise version itself cannot be $\NP{}$-hard as a language, because $\NP{}$ only contains non-promise problems.}; and deterministic versions of this reduction were obtained in~\cite{cheng2009deterministic} and~\cite{austrin2014simple}. We use the
    following consequence from the literature.
    
    \begin{theorem}[Theorem 1.1 of \cite{austrin2014simple}]
    \label{thm:known-gap}
    There is a $\lambda>1$ and a
    deterministic polynomial-time map $R$ such that, for every Boolean formula $\varphi$,
    $R(\varphi)=(G_\varphi,t_\varphi)$ is a promised instance of $\lambda$-\GapMWC such that
    \[
        \varphi \in \textsc{SAT} \implies \Delta(G_{\varphi})\leq t_{\varphi},
    \]
    and
    \[
        \varphi \notin \textsc{SAT} \implies \Delta(G_{\varphi})>\lambda t_{\varphi}.
    \]
    \end{theorem}
    
    From now on, we fix $\lambda$ as the value which exists thanks to this result. Theorem~\ref{thm:known-gap} implies that the existence of an algorithm solving $\lambda$-\GapMWC in polynomial time would imply that $\Pclass{} = \NP{}$. Similarly, the existence of a polynomial-time \textit{quantum} algorithm for this problem would imply that $\NP{} \subseteq \BQP{}$. We highlight this in the following corollary.
    
    \begin{corollary}
    \label{coro:hardness-of-code}
    If $\lambda$-\GapMWC$\in\BQP{}$, then $\NP \subseteq \BQP{}$.
    \end{corollary}
    
    \begin{proof}
    To solve \textsc{SAT} in \BQP{} we can apply the reduction from Theorem~\ref{thm:known-gap} and then the algorithm that solves $\lambda$-\GapMWC.
    \end{proof}
    
    Our strategy to prove that $\LPC \in \BQP{}$ implies $\NP{} \subseteq \BQP{}$ consists simply in reducing $\lambda$-\GapMWC to \LPC. Then, the existence of a quantum algorithm to solve \LPC would imply, through Corollary~\ref{coro:hardness-of-code}, that $\NP{} \subseteq \BQP{}$.
    
    \section{Basic lemmas}
    \label{sec:basic-lemmas}
    
    The circuit constructed by the reduction will be a randomized \textit{classical} circuit obtained by sampling the columns of $G$. The following lemma describes the basic properties of this type of construction.
    
    \begin{lemma}
    \label{lem:diagonal-pauli}
    Let $V$ be a random variable over $\F_2^n$, and let
    \[
        \rho_V:=\E_V[\ket{V}\bra{V}]
    \]
    be the corresponding diagonal density matrix. Then, for every $x\in\F_2^n$,
    \[
        \Tr(Z^x\rho_V)=\E_V[(-1)^{x\cdot V}].
    \]
    Moreover, if $P\in\cP_n$ contains at least one $X$ or $Y$, then
    \[
        \Tr(P\rho_V)=0.
    \]
    \end{lemma}
    
    \begin{proof}
    For any computational-basis string $y\in\F_2^n$, the vector $\ket{y}$ is an eigenvector of
    $Z^x$, with eigenvalue
    \[
        (-1)^{x_1y_1}\cdots(-1)^{x_ny_n}=(-1)^{x\cdot y}.
    \]
    Therefore
    \[
    \begin{aligned}
        \Tr(Z^x\rho_V)
          &=\E_V\bigl[\Tr(Z^x\ket{V}\bra{V})\bigr]  \\
          &=\E_V\bigl[\bra{V}Z^x\ket{V}\bigr]       \\
          &=\E_V[(-1)^{x\cdot V}].
    \end{aligned}
    \]
    Regarding the second statement, if $P$ contains an $X$ or a $Y$ then $P$ flips at least one computational-basis
    bit, and hence $P\ket{y}$ is orthogonal to $\ket{y}$, for every $y$. This implies that $\Tr(P\rho_V)=\E_V[\bra{V}P\ket{V}]=0$.
    \end{proof}
    
    The following lemma indicates how we can use expressions of the form $\E_V[(-1)^{x\cdot V}]$ to estimate the Hamming weight of a codeword.
    
    \begin{lemma}
    \label{lem:bias-weight}
    Let $G\in\F_2^{r\times m}$ and let $g_1,
    \ldots,g_m\in\F_2^r$ be its columns. Then, for every $x\in\F_2^r$,
    \[
        \sum_{j=1}^m(-1)^{x\cdot g_j}=m-2\wt(xG).
    \]
    More generally, if $\widetilde g_1,\ldots,\widetilde g_N$ are obtained by appending
    $N-m$ zero columns to $G$, then
    \[
        \sum_{j=1}^N(-1)^{x\cdot \widetilde g_j}=N-2\wt(xG).
    \]
    \end{lemma}
    
    \begin{proof}
    The $j$-th coordinate of the row-vector codeword $xG$ is
    \[
        (xG)_j=x\cdot g_j\pmod 2.
    \]
    Thus
    \[
        (-1)^{x\cdot g_j}=(-1)^{(xG)_j}.
    \]
    If $(xG)_j=0$, this contribution is $+1$, while if $(xG)_j=1$ this contribution is $-1$. The
    number of indices $j$ for which $(xG)_j=1$ is exactly $\wt(xG)$, and the number of indices for
    which $(xG)_j=0$ is $m-\wt(xG)$. Hence
    \[
    \begin{aligned}
        \sum_{j=1}^m(-1)^{x\cdot g_j}
          &=\#\{j:(xG)_j=0\}-\#\{j:(xG)_j=1\}\\
          &=(m-\wt(xG))-\wt(xG)\\
          &=m-2\wt(xG).
    \end{aligned}
    \]
    If $N-m$ zero columns are appended, each padded zero column contributes
    $(-1)^{x\cdot 0^r}=1$. Therefore the sum increases by $N-m$, giving
    \[
        (m-2\wt(xG))+(N-m)=N-2\wt(xG).
    \]
    \end{proof}
    
    For the proof of the pure state case we will require the existence of an unitary operator $U$ able to separate the states $\ket{0^c}$ and $\ket{10^{c-1}}$ in the sense that for all Pauli matrices $P$ over $c$ qubits it holds that $|\bra{x0^{c-1}} U^{\dagger}P U\ket{y0^{c-1}}|$ is bounded for every $x,y \in \{0, 1\}$. We use concentration arguments to justify the existence of such $U$. 
    
    \begin{lemma}
    \label{lem:pauli-expectation-concentration}
    There is a universal constant $c_0>0$ such that, for any fixed traceless Hermitian operator $A$ over $\mathbb C^D$ with $\|A\|\leq 1$, it holds that, if $\ket{\psi}$ is distributed according to the Haar measure on the unit sphere of $\mathbb C^D$, then
    \[
        \Pr\left[\left|\bra{\psi}A\ket{\psi}\right|>\delta\right]
          \leq 2\exp(-c_0D\delta^2)
    \]
    for every $\delta>0$.
    \end{lemma}
    
    \begin{proof}
    This is the standard form of Levy's lemma~\cite{ledoux2001concentration}. More precisely, the inequality can be obtained by applying Lemma 3 from~\cite{hayden2006aspects} with the function
    \[
        f(\psi):=\bra{\psi}A\ket{\psi}
    \]
    which has mean $\Tr(A)/D=0$ and a bounded Lipschitz constant.
    \end{proof}
    
    The previous result is intuitive: it states that for any traceless operator $A$ the states $\ket{\psi}$ and $A\ket{\psi}$ have small overlap for most choices of $\ket{\psi}$. We employ this lemma in the following result.
    
    \begin{lemma}
    \label{lem:pauli-flat-encoding}
    For every constant $\eta>0$, there exist a constant $c$ and an unitary $U$ acting on $(\mathbb{C}^2)^{\otimes c}$ such that
    \begin{align}
        |\bra{x 0^{c-1}}U^{\dagger}PU\ket{y0^{c-1}}| \leq \eta
    \end{align}
    for every $x,y \in \{0, 1\}$ and non-identity Pauli matrix $P$ on $c$ qubits.
    \end{lemma}
    
    \begin{proof}
    Write
    \[
        \ket{e_0}:=\ket{0^c},
        \qquad
        \ket{e_1}:=\ket{10^{c-1}}.
    \]
    Consider the following set of vectors in $(\mathbb C^2)^{\otimes c}$:
    \[
    \begin{aligned}
        \mathcal S:=\{&\ket{e_0},\ket{e_1},
        (\ket{e_0}+\ket{e_1})/\sqrt2,
        (\ket{e_0}-\ket{e_1})/\sqrt2,\\
        &(\ket{e_0}+i\ket{e_1})/\sqrt2,
        (\ket{e_0}-i\ket{e_1})/\sqrt2\}.
    \end{aligned}
    \]
    Let $U$ be a Haar-random unitary on $(\mathbb C^2)^{\otimes c}$. For every fixed vector
    $\ket{\alpha}\in\mathcal S$, the random vector $U\ket{\alpha}$ is Haar-distributed on the unit sphere. Hence, for every fixed non-identity Pauli string $P$ on $c$ qubits, Lemma~\ref{lem:pauli-expectation-concentration} applied to the random vector $U\ket{\alpha}$ gives
    \[
        \Pr{}_U\left[
        \left|\bra{\alpha}U^\dagger P U\ket{\alpha}\right|>\eta/2
        \right]
          \leq 2\exp(-c_0\eta^2 2^c/4).
    \]
    We now take a union bound over the six fixed vectors in $\mathcal S$ and over the $4^c-1$ non-identity Pauli strings over $c$ qubits. The probability that at least one of the inequalities fails is at most $6(4^c-1)\cdot 2\exp(-c_0\eta^2 2^c/4)$, and for a sufficiently large constant $c$ this quantity is strictly smaller than one. Thus, there exists an unitary $U$ such that, simultaneously for every $\ket{\alpha}\in\mathcal S$ and every non-identity Pauli string $P$,
    \[
        \left|\bra{\alpha}U^\dagger P U\ket{\alpha}\right|\leq \eta/2.
    \]
    Fix such an unitary $U$. By definition, it holds that
    \[
        \left|\bra{e_a}U^\dagger P U\ket{e_a}\right|
          \leq \eta/2<\eta
    \]
    for $a \in \{0, 1\}$. For the non-diagonal terms, define
    \[
        \ket{\overline 0}:=U\ket{e_0},
        \qquad
        \ket{\overline 1}:=U\ket{e_1}.
    \]
    By linearity of $U$, the fixed vectors
    $(\ket{e_0}\pm\ket{e_1})/\sqrt2$ are mapped to
    \[
        \ket{\overline \pm}:=\frac{\ket{\overline 0}\pm\ket{\overline 1}}{\sqrt2},
    \]
    and the fixed vectors $(\ket{e_0}\pm i\ket{e_1})/\sqrt2$ are mapped to
    \[
        \ket{\overline{\pm i}}:=\frac{\ket{\overline 0}\pm i\ket{\overline 1}}{\sqrt2}.
    \]
    Thus the inequalities already established guarantee that
    \[
        \left|\bra{\overline +}P\ket{\overline +}\right|,
        \left|\bra{\overline -}P\ket{\overline -}\right|,
        \left|\bra{\overline{+i}}P\ket{\overline{+i}}\right|,
        \left|\bra{\overline{-i}}P\ket{\overline{-i}}\right|
        \leq \eta/2.
    \]
    for every non-identity Pauli matrix $P$. Using the polarization identities we can conclude that
    \[
        \left|\bra{\overline 0}P\ket{\overline 1}\right|
          \leq \eta.
    \]
    \end{proof}

    \section{Main results}
    \label{sec:main-results}
    
    We first state the simple observation that \LPC belongs to \QCMA{}.
    
    \begin{proposition}\label{prop:qcma}
    If $\varepsilon = \Omega(1/\poly(|C|))$, then \LPC is in \QCMA{}.
    \end{proposition}
    
    \begin{proof}
    If $(C, \varepsilon)$ is a positive instance of \LPC, then a classical witness for this fact is the Pauli matrix $P\in\cP_n\setminus\{I^{\otimes n}\}$ such that $|\Tr(P\rho_C)|\geq \varepsilon$. The verifier estimates the observable $P$ on copies of $\rho_C$ and checks that the absolute value of the estimate is above a threshold, say $3\varepsilon/4$. Since each Pauli measurement outcome lies in $\{-1,+1\}$, Hoeffding's inequality implies that $O(1/\varepsilon^2)$ repetitions are enough to estimate $\Tr(P\rho_C)$ to additive error at most $\varepsilon/4$ with probability at least $2/3$. 
    
    Therefore, if indeed $|\Tr(P\rho_C)|\geq\varepsilon$, the verifier accepts with probability at least $2/3$. If the instance is negative, then every possible witness satisfies $|\Tr(P\rho_C)|\leq\varepsilon/2$, and the same verifier rejects with probability at least $2/3$. If $\varepsilon = \Omega(1 / \poly(|C|))$, the whole algorithm can be implemented in polynomial time in the size of $C$.
    \end{proof}
    
    We now describe our main reduction from $\lambda$-\GapMWC to \LPC. Let $(G,t)$ be an instance of $\lambda$-\GapMWC, where
    $G\in\F_2^{r\times m}$. As stated before, we can assume without loss of generality that $r \geq 1$, that $G$ is full row rank, and that $1 \leq t < m$.
    
    Choose a power of two $N$ such that
    \[
        N>2\lambda m.
    \]
    For example, we can take
    $N=2^{\lceil\log_2(2\lambda m+1)\rceil}$. Then $N=O(m)$ and
    \[
        0<\frac{2t}{N}<\frac{2\lambda t}{N}<1.
    \]
    Let $g_1,\ldots,g_m\in\F_2^r$ be the columns of $G$. Pad columns
    $\widetilde g_1,\ldots,\widetilde g_N\in\F_2^r$ as
    \[
        \widetilde g_j:=
        \begin{cases}
          g_j, & 1\leq j\leq m,\\
          0^r, & m<j\leq N.
        \end{cases}
    \]
    We will consider the random variable $V_1\in\F_2^r$ obtained by uniformly picking a column from this matrix, i.e.
    \[
        V_1:=\widetilde g_J,
        \qquad
        J\sim\mathrm{Unif}\{1,
        \ldots,N\}.
    \]
    Let $V^{(1)},\ldots,V^{(k)}$ be independent copies of $V_1$, and consider the direct sum over $\F_2$
    \[
        V:=V^{(1)}\oplus\cdots\oplus V^{(k)}.
    \]
    Our reduction will map the matrix $G$ to the state
    \[
        \rho_{G,t}:=\E_V[\ket{V}\bra{V}].
    \]
    
    It remains to choose $k$ and $\varepsilon$. Let
    \[
        a:=\frac{2t}{N},
        \qquad
        b:=\frac{2\lambda t}{N}.
    \]
    Then $0<a<b<1$. Set
    \begin{equation}
    \label{eq:k-choice-v2}
        k:=\left\lceil
           \frac{\ln 2}{\ln\left(\frac{1-a}{1-b}\right)}
        \right\rceil,
    \end{equation}
    and pick as threshold
    \begin{equation}
    \label{eq:epsilon-choice-v2}
        \varepsilon:=(1-a)^k=\left(1-\frac{2t}{N}\right)^k.
    \end{equation}
    The following holds.
    
    \begin{lemma}
    \label{lem:poly-v2}
    The integer $k$ is $O(m)$ and $\varepsilon = \Omega(1)$. Also, the circuit preparing
    $\rho_{G,t}$ has a polynomial-size description.
    \end{lemma}
    
    \begin{proof}
    Since $0<a<b<1$,
    \[
        \ln\left(\frac{1-a}{1-b}\right)
          =\ln(1-a)-\ln(1-b)
          =\int_a^b\frac{du}{1-u}.
    \]
    The integrand is at least $1$ on $[a,b]$, so
    \[
        \ln\left(\frac{1-a}{1-b}\right)\geq b-a=\frac{2(\lambda-1)t}{N}.
    \]
    Therefore
    \[
        k\leq \frac{N\ln2}{2(\lambda-1)t}+1.
    \]
    Since $t\geq1$ and $N=O(m)$, we conclude that $k=O(m)$.
    
    The circuit that prepares $\rho_{G,t}$ uses $k \log_2 N$ qubits to sample $k$ columns of the padded version of $G$. More precisely, for each $i\in\{1,\ldots,k\}$ we prepare a register $J_i$ of $\log_2N$ qubits in the uniform superposition over $\{1,\ldots,N\}$ by applying Hadamard gates. Then, we initialize an $r$-qubit output register to $0^r$ and compute
    \[
        (J_1,\ldots,J_k,0^r)
        \longmapsto
        \left(J_1,\ldots,J_k,\widetilde g_{J_1}\oplus\cdots\oplus\widetilde g_{J_k}\right).
    \]
    The columns $\widetilde g_j$ can be hardwired to make this operation efficient. The whole circuit has polynomial size because $k=O(m)$ and $N=O(m)$. After this computation, the state before discarding the $J_i$ registers is
    \[
        \frac{1}{\sqrt{N^k}}
        \sum_{j_1,\ldots,j_k\in\{1,\ldots,N\}}
        \ket{j_1,\ldots,j_k}\ket{\widetilde g_{j_1}\oplus\cdots\oplus\widetilde g_{j_k}}.
    \]
    Discarding the $J_1,\ldots,J_k$ registers leaves exactly the classical mixture
    \[
        \rho_{G,t}=
        \frac{1}{N^k}
        \sum_{j_1,\ldots,j_k\in\{1,\ldots,N\}}
        \ket{\widetilde g_{j_1}\oplus\cdots\oplus\widetilde g_{j_k}}
        \bra{\widetilde g_{j_1}\oplus\cdots\oplus\widetilde g_{j_k}}.
    \]
    No approximations are introduced because we picked $N$ as a power of two.
    
    Finally, we lower bound $\varepsilon$. Since $N>2\lambda m$ and $t<m$, we have
    $a=2t/N<1/\lambda$. Also, from Eq.~\eqref{eq:k-choice-v2} and the bound above, we have
    \[
        k\leq \frac{\ln 2}{(\lambda-1)a}+1.
    \]
    Using $-\ln(1-a)\leq a/(1-a)$ and $a<1/\lambda$, we get
    \[
    \begin{aligned}
        -\ln\varepsilon
          &=-k\ln(1-a)\\
          &\leq \left(\frac{\ln 2}{(\lambda-1)a}+1\right)\frac{a}{1-a}\\
          &\leq \frac{\lambda\ln2}{(\lambda-1)^2}+\frac{1}{\lambda-1}.
    \end{aligned}
    \]
    In the last step we used that $a < 1/\lambda$ implies $1/(1-a) < \lambda/(\lambda-1)$ and $a/(1-a)<1/(\lambda-1)$. Thus
    \[
        \varepsilon\geq
        \exp\left(-\frac{\lambda\ln2}{(\lambda-1)^2}-\frac{1}{\lambda-1}\right)
        = \Omega(1),
    \]
    since $\lambda$ is fixed.
    \end{proof}
    
    The next lemma justifies the correctness of the reduction.
    
    \begin{lemma}
    \label{lem:key-identity-v2}
    For every $x\in\F_2^r$,
    \[
        \Tr(Z^x\rho_{G,t})
          =\left(1-\frac{2\wt(xG)}{N}\right)^k.
    \]
    Moreover, every Pauli string containing at least one $X$ or $Y$ has expectation zero on
    $\rho_{G,t}$.
    \end{lemma}
    
    \begin{proof}
    By the definition of $V$, the density matrix is
    \[
        \rho_{G,t}=
        \frac{1}{N^k}
        \sum_{j_1,\ldots,j_k\in\{1,\ldots,N\}}
        \ket{\widetilde g_{j_1}\oplus\cdots\oplus\widetilde g_{j_k}}
        \bra{\widetilde g_{j_1}\oplus\cdots\oplus\widetilde g_{j_k}}.
    \]
    Therefore, using Lemma~\ref{lem:diagonal-pauli},
    \[
    \begin{aligned}
        \Tr(Z^x\rho_{G,t})
          &=\frac{1}{N^k}
            \sum_{j_1,\ldots,j_k\in\{1,\ldots,N\}}
            (-1)^{x\cdot(\widetilde g_{j_1}\oplus\cdots\oplus\widetilde g_{j_k})}.
    \end{aligned}
    \]
    It holds that
    \[
        x\cdot(\widetilde g_{j_1}\oplus\cdots\oplus\widetilde g_{j_k})
          =x\cdot\widetilde g_{j_1}+\cdots+x\cdot\widetilde g_{j_k}
          \pmod 2.
    \]
    Thus
    \[
        (-1)^{x\cdot(\widetilde g_{j_1}\oplus\cdots\oplus\widetilde g_{j_k})}
          =\prod_{i=1}^k(-1)^{x\cdot\widetilde g_{j_i}}.
    \]
    We now isolate the first sampled column. Namely,
    \[
    \begin{aligned}
        \Tr(Z^x\rho_{G,t})
          &=\frac{1}{N^k}
            \sum_{j_1,\ldots,j_k\in\{1,\ldots,N\}}
            \prod_{i=1}^k(-1)^{x\cdot\widetilde g_{j_i}}\\
          &=\left(\frac1N\sum_{j_1=1}^N(-1)^{x\cdot\widetilde g_{j_1}}\right)
            \left(
            \frac{1}{N^{k-1}}
            \sum_{j_2,\ldots,j_k\in\{1,\ldots,N\}}
            \prod_{i=2}^k(-1)^{x\cdot\widetilde g_{j_i}}
            \right).
    \end{aligned}
    \]
    The first factor is exactly the expectation of the same observable for one elementary sample. Indeed, if
    \[
        \rho_{V_1}:=\E_{V_1}[\ket{V_1}\bra{V_1}],
    \]
    then
    \[
    \begin{aligned}
        \Tr(Z^x\rho_{V_1})
          &=\E_{V_1}[(-1)^{x\cdot V_1}]\\
          &=\frac1N\sum_{j=1}^N(-1)^{x\cdot\widetilde g_j}.
    \end{aligned}
    \]
    The second factor has the same form as the original expression, but with the remaining $k-1$ samples. Thus the same isolation can be repeated for $j_2$, then for $j_3$, and so on. Iterating this identity gives
    \[
        \Tr(Z^x\rho_{G,t})
          =\left(\frac1N\sum_{j=1}^N(-1)^{x\cdot\widetilde g_j}\right)^k.
    \]
    For the elementary sample $V_1=\widetilde g_J$, Lemma~\ref{lem:bias-weight} gives
    \[
    \begin{aligned}
        \frac1N\sum_{j=1}^N(-1)^{x\cdot\widetilde g_j}
          &=\frac{N-2\wt(xG)}{N}\\
          &=1-\frac{2\wt(xG)}{N}.
    \end{aligned}
    \]
    Therefore
    \[
        \Tr(Z^x\rho_{G,t})
          =\left(1-\frac{2\wt(xG)}{N}\right)^k.
    \]
    This implies the desired equality. The final statement from this lemma follows as well from
    Lemma~\ref{lem:diagonal-pauli}.
    \end{proof}
    
    We now combine all these lemmas into a single theorem.
    
    \begin{theorem}
    \label{thm:reduction-v2}
    The map described above is a deterministic classical polynomial-time reduction from $\lambda$-\GapMWC to \LPC. The produced \LPC
    instances use diagonal states corresponding to randomized classical circuits, and $\varepsilon$ is bounded from below by a constant.
    \end{theorem}
    
    \begin{proof}
    Lemma~\ref{lem:poly-v2} states that the reduction can be implemented in polynomial time and that $\varepsilon = \Omega(1)$.
    
    For the correctness, suppose $(G,t)$ is a yes-instance of $\lambda$-\GapMWC. Then $\Delta(G)\leq t$, so there is
    a nonzero codeword $xG$ with $\wt(xG)\leq t$. Hence $Z^x$ is
    a non-identity Pauli, and by Lemma~\ref{lem:key-identity-v2} we have
    \[
        \Tr(Z^x\rho_{G,t})
          =\left(1-\frac{2\wt(xG)}{N}\right)^k
          \geq\left(1-\frac{2t}{N}\right)^k
          =\varepsilon.
    \]
    Thus, the produced \LPC instance is a yes-instance.
    
    Now suppose, on the contrary, that $(G,t)$ is a no-instance of $\lambda$-\GapMWC. Then every nonzero codeword has
    weight greater than $\lambda t$. Since $G$ has full row rank, every nonzero
    $x\in\F_2^r$ gives a nonzero codeword $xG$, and so
    \[
        \wt(xG)>\lambda t.
    \]
    Then, by Lemma~\ref{lem:key-identity-v2},
    \[
        \Tr(Z^x\rho_{G,t})
          =\left(1-\frac{2\wt(xG)}{N}\right)^k
          <\left(1-\frac{2\lambda t}{N}\right)^k
    \]
    for every nonzero $x$. The choice of $k$ in Eq.~\eqref{eq:k-choice-v2} implies
    \[
        k\ln\left(\frac{1-2t/N}{1-2\lambda t/N}\right)\geq\ln2.
    \]
    Equivalently,
    \[
        \left(1-\frac{2\lambda t}{N}\right)^k
          \leq\frac12\left(1-\frac{2t}{N}\right)^k
          =\frac{\varepsilon}{2}.
    \]
    Therefore every non-identity $Z$-type Pauli has absolute expectation strictly smaller than $\varepsilon/2$. Every Pauli containing an $X$ or a $Y$ has expectation zero by Lemma~\ref{lem:key-identity-v2}. Hence the produced \LPC instance is a no-instance.
    \end{proof}
    
    The next corollary is therefore our main result.
    
    \begin{corollary}
    \label{coro:main}
    If $\LPC \in \BQP{}$, then $\mathrm{NP}\subseteq\mathrm{BQP}$.
    \end{corollary}
    
    \begin{proof}
    Combine the reduction from Theorem~\ref{thm:reduction-v2} with Corollary~\ref{coro:hardness-of-code}.
    \end{proof}
    
    Given this result, the question remains whether $\LPC$ is in $\NP{}$ or rather is $\QCMA{}$-hard. We provide some insight into this question by proving that it is \BQP{}-hard. This implies that, under the hypothesis that $\BQP \nsubseteq \NP$\footnote{As mentioned before, due to the presence of promises in \BQP{} problems, this statement shouldn't be read as an exact inclusion, but rather as ``There are some problems in \BQP{} that cannot be solved in non-deterministic polynomial time''.}, then $\LPC \notin \NP$.
    
    \begin{proposition}\label{prop:bqp-hard}
        \LPC is \BQP-hard.
    \end{proposition}
    
    \begin{proof}
    Let $\Pi\in\BQP{}$ and let $\{C_n\}_{n \in \mathbb{N}}$ be the sequence of quantum circuits that solves $\Pi$ by measuring the first qubit of the circuit, which we refer to as $q_0$ from now on. We show how to construct, for every $x \in \{0,1\}^n$, a state preparation circuit whose output register consists of a one qubit state $\rho_x$ such that $C_n$ accepts $\ket{x}$ with probability above 2/3 if and only if there is a non-identity Pauli $P$ such that $|\Tr{}[P\rho_x]| \geq 2/3$.
    
    Assume that $C_n$ works over $m \geq n$ qubits and defers all measurements to the end, and append to it two qubits $q_1$ and $q_2$ initialized in $\ket{0}$. Add a Hadamard gate acting on $q_1$, and after the gates from $C_n$ apply the operation $U$ defined as
    \begin{align*}
        \ket{b_0,b_1,b_2} \overset{U}{\mapsto} \ket{b_0, b_1, b_2 \oplus (1-b_0)b_1}
    \end{align*}
    to the qubits $q_0$, $q_1$ and $q_2$. We can write this gate as
    \begin{align*}
        U = \ket{0}\bra{0}\otimes \textsc{CNOT}_{b_1, b_2} + \ket{1}\bra{1} \otimes I^{2}    
    \end{align*}
    Then, trace out all the qubits except for $q_2$, and let $\rho_x$ be this state. We show that, if $p_x = \Tr{}[(\ket{1}\bra{1} \otimes I^{m-1})C_n\ket{x0^{m-n}}\bra{x0^{m-n}}C_n^{\dagger}]$ is the accepting probability of $C_n\ket{x0^{m-n}}$, then $\rho_x = p_x \ket{0}\bra{0} + (1-p_x)I/2$. The correctness of the reduction follows from this, since then $\Tr(X\rho_x) = \Tr(Y\rho_x) = 0$ and $\Tr(Z\rho_x) = p_x$.
    
    Write $C_n\ket{x0^{m-n}} = \ket{0} \ket{\sigma_0} + \ket{1}\ket{\sigma_1}$ where $||\ket{\sigma_1}||^2 = p_x$ and $||\ket{\sigma_0}||^2 = 1-p_x$. It holds that
    \begin{align*}
        C_n\ket{x0^{m-n}} \ket{+} \ket{0} \overset{U}{\mapsto} \ket{0}\ket{\sigma_0} \left(\frac{1}{\sqrt{2}}(\ket{00} + \ket{11})\right) + \ket{1}\ket{\sigma_1} \ket{+}\ket{0}
    \end{align*}
    Thus, if we trace all the qubits except the last one, we obtain
    \begin{align*}
        \rho_x = (1-p_x)I/2 + p_x \ket{0}\bra{0}
    \end{align*}
    where we used the fact that, if $\ket{\Phi} = \frac{1}{\sqrt{2}}(\ket{00} + \ket{11})$, tracing over the first qubit we get $\Tr{}_1[\ket{\Phi}\bra{\Phi}] = I/2$. 
    \end{proof}

    \section{The pure state case}
    \label{sec:pure state-case}
    
    We now show that the same reduction can be improved to prove hardness for pure states. To do this, we employ the same reduction from Theorem~\ref{thm:reduction-v2} but use the unitary from Lemma~\ref{lem:pauli-flat-encoding} to ``hide'' the undesirable Paulis which do not correspond to measuring the qubits containing the superposition of the columns of $\widetilde{G}$. 
    
    \begin{proposition}
    \label{prop:pure-np-hard}
    If \LPC restricted to pure states (i.e. no qubits are discarded from the circuit) is in $\BQP{}$, then $\NP{} \subseteq \BQP{}$.
    \end{proposition}
    
    \begin{proof}
    We modify the reduction from Theorem~\ref{thm:reduction-v2}. Let $\varepsilon_0>0$ be the constant lower bound on $\varepsilon$ obtained in Lemma~\ref{lem:poly-v2}. Fix a constant $\eta>0$ such that
    \[
        2\eta\leq \varepsilon_0/2,
    \]
    and let $U$ be the unitary from Lemma~\ref{lem:pauli-flat-encoding} picking $\eta$ as constant. By universality we can assume that we have a circuit approximating $U$ up to any precision, and in particular we may as well simply assume access to a circuit $C_U$ such that\footnote{Being more precise, since $C_U$ may not be exactly $U$, to obtain the bound in the equation we can define $U$ by picking as constant $\eta /2$ and then take $C_U$ as a sufficiently precise approximation of $U$ in operator norm.}
    \begin{align*}
        |\bra{x0^{c-1}} C_U^\dagger PC_U\ket{y0^{c-1}}| \leq \eta
    \end{align*}
    for $x,y \in \{0,1\}$ and every non-identity Pauli matrix $P$ on $c$ qbits. 
    
    As before, let $N$ be a power of two, and let $L:=k\log_2 N$. We identify each tuple
    $(j_1,\ldots,j_k)\in\{1,\ldots,N\}^k$ with a binary string $s\in\F_2^L$. Let
    \[
        v(s):=\widetilde g_{j_1}\oplus\cdots\oplus\widetilde g_{j_k}\in\F_2^r.
    \]
    For $s=(s_1,\ldots,s_L)$, define
    \begin{align*}
        \ket{\overline{s_i}} = C_U\ket{s_i0^{c-1}}.
    \end{align*}
    and
    \[
        \ket{\overline s}:=\bigotimes_{i=1}^L\ket{\overline{s_i}}.
    \]
    where we appended $L(c-1)$ qubits to the system.
    The pure state reduction outputs the state
    \[
        \ket{\Psi_{G,t}}
          :=\frac{1}{\sqrt{2^L}}
            \sum_{s\in\F_2^L}
            \ket{v(s)}\ket{\overline s}.
    \]
    This state is prepared in polynomial time: the only difference from the previous circuit is the inclusion of the $L$ $C_U$ gates, which have constant size (they are fixed once we fix $\eta$).
    
    If a Pauli operator acts trivially on the register containing $\overline{s}$, then its expectation on $\ket{\Psi_{G,t}}$ is exactly the same as the expectation on the mixed state $\rho_{G,t}$. Indeed, for every Pauli $P_A$ on the first register,
    \[
    \begin{aligned}
        \bra{\Psi_{G,t}} P_A\otimes I \ket{\Psi_{G,t}}
          &=\frac{1}{2^L}
            \sum_{s,s'\in\F_2^L}
            \bra{v(s')}P_A\ket{v(s)}\braket{\overline{s'}}{\overline{s}}\\
          &=\frac{1}{2^L}
            \sum_{s\in\F_2^L}
            \bra{v(s)}P_A\ket{v(s)}\\
          &=\Tr(P_A\rho_{G,t}).
    \end{aligned}
    \]
    In particular, for every $x\in\F_2^r$,
    \[
        \bra{\Psi_{G,t}} Z^x\otimes I \ket{\Psi_{G,t}}
          =\left(1-\frac{2\wt(xG)}{N}\right)^k.
    \]
    
    It remains to check that the encoded sample register does not create additional large Pauli coefficients. Let $P_A\otimes P_B$ be a Pauli string, where $P_A$ acts on the first register and $P_B$ acts on the encoded sample register. Write
    \[
        P_B=P_1\otimes\cdots\otimes P_L,
    \]
    where each $P_i$ is a Pauli string on the $c$ physical qubits encoding the $i$-th sample bit. Then
    \[
    \begin{aligned}
        \bra{\Psi_{G,t}}P_A\otimes P_B\ket{\Psi_{G,t}}
          &=\frac{1}{2^L}
            \sum_{s,s'\in\F_2^L}
            \bra{v(s')}P_A\ket{v(s)}
            \prod_{i=1}^L
            \bra{\overline{s_i'}}P_i\ket{\overline{s_i}}.
    \end{aligned}
    \]
    Since $|\bra{v(s')}P_A\ket{v(s)}|\leq 1$, we get
    \[
    \begin{aligned}
        \left|\bra{\Psi_{G,t}}P_A\otimes P_B\ket{\Psi_{G,t}}\right|
          &\leq
            \frac{1}{2^L}
            \sum_{s,s'\in\F_2^L}
            \prod_{i=1}^L
            \left|\bra{\overline{s_i'}}P_i\ket{\overline{s_i}}\right|\\
          &=
            \prod_{i=1}^L
            \left(
            \frac12
            \sum_{x,y\in\{0,1\}}
            \left|\bra{x0^{c-1}}C_U^\dagger P_i C_U\ket{\overline y0^{c-1}}\right|
            \right).
    \end{aligned}
    \]
    If $P_i=I^{c}$, then the corresponding factor is equal to $1$, because
    $\braket{x0^{c-1}}{y0^{c-1}}=0$ for $x\neq y$ and is equal to $1$ for $x=y$. If $P_i\neq I^{c}$, then Lemma~\ref{lem:pauli-flat-encoding} gives
    \[
        \frac12
        \sum_{x,y\in\{0,1\}}
        \left|\bra{x0^{c-1}}C_U^\dagger P_i C_U\ket{y0^{c-1}}\right|
          \leq 2\eta.
    \]
    Therefore, if $P_B$ is non-identity, then at least one of the $P_i$'s is non-identity, and
    \[
        \left|\bra{\Psi_{G,t}}P_A\otimes P_B\ket{\Psi_{G,t}}\right|
          \leq 2\eta
          \leq \varepsilon_0/2
          \leq \varepsilon/2.
    \]
    
    Now suppose $(G,t)$ is a yes-instance of $\lambda$-\GapMWC. By the proof of Theorem~\ref{thm:reduction-v2}, there is a nonzero $x\in\F_2^r$ such that
    \[
        \left|\bra{\Psi_{G,t}} Z^x\otimes I \ket{\Psi_{G,t}}\right|
          \geq \varepsilon.
    \]
    Thus the pure state instance is a yes-instance.
    
    On the other hand, suppose $(G,t)$ is a no-instance. If a non-identity Pauli acts trivially on the encoded sample register, then the preceding equality with $\rho_{G,t}$ and Theorem~\ref{thm:reduction-v2} imply that its absolute expectation is at most $\varepsilon/2$. If it acts nontrivially on the encoded sample register, then the bound above gives absolute expectation at most $\varepsilon/2$. Hence every non-identity Pauli has absolute expectation at most $\varepsilon/2$, and the pure state instance is a no-instance.
    
    This gives a deterministic polynomial-time reduction from $\lambda$-\GapMWC to the pure state restriction of \LPC. Combining this with Corollary~\ref{coro:hardness-of-code} proves the claim.
    \end{proof}
    
    \section{Conclusion}
    \label{sec:discussion}

    In this paper we introduced the problem of detecting, given a mechanism to prepare a state $\rho$ and a threshold $\varepsilon$, whether there exists some non-identity Pauli matrix $P$ such that $|\Tr{}(P\rho)|\geq \varepsilon$. Through a reduction from the gapped version of the minimum weight code problem we showed that the existence of a $\BQP{}$ algorithm to solve this problem would imply that $\NP{} \subseteq \BQP{}$, even when restricted to pure states and a constant threshold. Moreover, we also proved that the problem is $\BQP{}$-hard.  Conceptually, these results indicate that under the usual complexity-theoretical assumption there is no efficient algorithm able to find (or even detect) large coefficients in the Pauli basis. 
    
    Note that we do not prove a tight bound, since we only have membership in \QCMA{} and hardness for $\NP{}$. In particular, we find it likely that \LPC is \QCMA-hard. Nonetheless, this result is already enough to clarify the situation regarding efficient algorithms for selective Pauli tomography: under the common assumptions an algorithm intended to find the largest Pauli coefficients with a polynomial amount of resources \textit{must} exploit additional structure, such as Pauli sparsity, stabilizer-like behavior~\cite{montanaro2017learning,leone2024learning,grewal2025efficient}, low rank~\cite{gross2010quantum}, locality, low entanglement~\cite{cramer2010efficient}, etc. Moreover, the \BQP{}-hardness result also suggests that a purely classical witness characterization of \LPC is unlikely, since that would imply that $\BQP{} \subseteq \NP{}$.

    \bibliographystyle{plain}
    \bibliography{biblio}

@article{vardy1997intractability,
  title={The intractability of computing the minimum distance of a code},
  author={Vardy, Alexander},
  journal={IEEE Transactions on Information Theory},
  volume={43},
  number={6},
  pages={1757--1766},
  year={1997},
  publisher={IEEE}
}

@article{hinsche2025efficient,
  title={Efficient distributed inner-product estimation via Pauli sampling},
  author={Hinsche, Marcel and Ioannou, Marios and Jerbi, Sofiene and Leone, Lorenzo and Eisert, Jens and Carrasco, Jose},
  journal={PRX Quantum},
  volume={6},
  number={3},
  pages={030354},
  year={2025},
  publisher={APS}
}

@article{gross2010quantum,
  title={Quantum state tomography via compressed sensing},
  author={Gross, David and Liu, Yi-Kai and Flammia, Steven T and Becker, Stephen and Eisert, Jens},
  journal={Physical review letters},
  volume={105},
  number={15},
  pages={150401},
  year={2010},
  publisher={APS}
}

@article{dumer2003hardness,
  title={Hardness of approximating the minimum distance of a linear code},
  author={Dumer, Ilya and Micciancio, Daniele and Sudan, Madhu},
  journal={IEEE Transactions on Information Theory},
  volume={49},
  number={1},
  pages={22--37},
  year={2003},
  publisher={IEEE}
}

@article{ciancaglini2026measuring,
  title={Measuring the largest coefficients of a quantum state},
  author={Ciancaglini, Nicol{\'a}s and Cifuentes, Santiago and Bellomo, Guido and Figueira, Santiago and Bendersky, Ariel},
  journal={arXiv preprint arXiv:2605.00341},
  year={2026}
}

@inproceedings{cheng2009deterministic,
  title={A deterministic reduction for the gap minimum distance problem},
  author={Cheng, Qi and Wan, Daqing},
  booktitle={Proceedings of the forty-first annual ACM symposium on Theory of computing},
  pages={33--38},
  year={2009}
}

@article{bendersky2009general,
  title={General theory of measurement with two copies of a quantum state},
  author={Bendersky, Ariel and Paz, Juan Pablo and Cunha, Marcelo Terra},
  journal={Physical review letters},
  volume={103},
  number={4},
  pages={040404},
  year={2009},
  publisher={APS}
}

@article{bendersky2008selective,
  title={Selective and efficient estimation of parameters for quantum process tomography},
  author={Bendersky, Ariel and Pastawski, Fernando and Paz, Juan Pablo},
  journal={Physical review letters},
  volume={100},
  number={19},
  pages={190403},
  year={2008},
  publisher={APS}
}

@article{bendersky2013selective,
  title={Selective and efficient quantum state tomography and its application to quantum process tomography},
  author={Bendersky, Ariel and Paz, Juan Pablo},
  journal={Physical Review A—Atomic, Molecular, and Optical Physics},
  volume={87},
  number={1},
  pages={012122},
  year={2013},
  publisher={APS}
}

@article{austrin2014simple,
  title={A simple deterministic reduction for the gap minimum distance of code problem},
  author={Austrin, Per and Khot, Subhash},
  journal={IEEE Transactions on Information Theory},
  volume={60},
  number={10},
  pages={6636--6645},
  year={2014},
  publisher={IEEE}
}

@book{ledoux2001concentration,
  title={The concentration of measure phenomenon},
  author={Ledoux, Michel},
  number={89},
  year={2001},
  publisher={American Mathematical Soc.}
}

@article{hayden2006aspects,
  title={Aspects of generic entanglement},
  author={Hayden, Patrick and Leung, Debbie W and Winter, Andreas},
  journal={Communications in mathematical physics},
  volume={265},
  number={1},
  pages={95--117},
  year={2006},
  publisher={Springer}
}

@article{huang2021efficient,
  title={Efficient estimation of Pauli observables by derandomization},
  author={Huang, Hsin-Yuan and Kueng, Richard and Preskill, John},
  journal={Physical review letters},
  volume={127},
  number={3},
  pages={030503},
  year={2021},
  publisher={APS}
}

@article{hangleiter2024bell,
  title={Bell sampling from quantum circuits},
  author={Hangleiter, Dominik and Gullans, Michael J},
  journal={Physical Review Letters},
  volume={133},
  number={2},
  pages={020601},
  year={2024},
  publisher={APS}
}

@inproceedings{aaronson2018shadow,
  title={Shadow tomography of quantum states},
  author={Aaronson, Scott},
  booktitle={Proceedings of the 50th annual ACM SIGACT symposium on theory of computing},
  pages={325--338},
  year={2018}
}

@article{montanaro2017learning,
  title={Learning stabilizer states by Bell sampling},
  author={Montanaro, Ashley},
  journal={arXiv preprint arXiv:1707.04012},
  year={2017}
}

@article{leone2024learning,
  title={Learning t-doped stabilizer states},
  author={Leone, Lorenzo and Oliviero, Salvatore FE and Hamma, Alioscia},
  journal={Quantum},
  volume={8},
  pages={1361},
  year={2024},
  publisher={Verein zur F{\"o}rderung des Open Access Publizierens in den Quantenwissenschaften}
}

@article{grewal2025efficient,
  title={Efficient learning of quantum states prepared with few non-Clifford gates},
  author={Grewal, Sabee and Iyer, Vishnu and Kretschmer, William and Liang, Daniel},
  journal={Quantum},
  volume={9},
  pages={1907},
  year={2025},
  publisher={Verein zur F{\"o}rderung des Open Access Publizierens in den Quantenwissenschaften}
}

@article{huang2020predicting,
  title={Predicting many properties of a quantum system from very few measurements},
  author={Huang, Hsin-Yuan and Kueng, Richard and Preskill, John},
  journal={Nature Physics},
  volume={16},
  number={10},
  pages={1050--1057},
  year={2020},
  publisher={Nature Publishing Group UK London}
}

@article{cramer2010efficient,
  title={Efficient quantum state tomography},
  author={Cramer, Marcus and Plenio, Martin B and Flammia, Steven T and Somma, Rolando and Gross, David and Bartlett, Stephen D and Landon-Cardinal, Olivier and Poulin, David and Liu, Yi-Kai},
  journal={Nature communications},
  volume={1},
  number={1},
  pages={149},
  year={2010},
  publisher={Nature Publishing Group UK London}
}
    
    \end{document}